\documentclass[11pt,letterpaper]{article}
\usepackage[utf8]{inputenc}
\usepackage[margin=1in]{geometry}
\usepackage{amsmath,amssymb,amsthm}
\usepackage[boxed, ruled, linesnumbered]{algorithm2e}
\usepackage{xcolor}
\usepackage{mathpazo}
 \usepackage{hyperref}
 \usepackage{graphicx}
\hypersetup{colorlinks=true,linkcolor=black,citecolor=[rgb]{0.1,0.1,0.9}}
\newtheorem{theorem}{Theorem}[section]
\newtheorem{claim}[theorem]{Claim}

\newtheorem{obs}[theorem]{Observation}

\newcommand{\Id}{\pi^{\textsf{Id}}}

\newcommand{\poly}{\mathrm{poly}}

\newcommand{\w}{\textbf{w}}
\newcommand{\x}{{x}}
\SetFuncSty{textsf}
\SetKwInput{KwInput}{Input}   
\SetKwFunction{FShuffle}{shuffle\(_i\)}
\SetKwFunction{FGenPerm}{GeneratePermutation}
\SetKwProg{Fn}{Function}{:}{\KwRet}

\title{{Explicit Good Codes Approaching  Distance 1 in Ulam Metric }}
\author{ Elazar Goldenberg\\The  Academic College of Tel Aviv-Yaffo\\ \texttt{elazargo@mta.ac.il}\\ \and Mursalin Habib\\  Rutgers University\\ \texttt{mursalin.habib@rutgers.edu}\and Karthik C.\ S.\footnote{This work was supported by the National Science Foundation under Grant CCF-2313372 and by the Simons  
 Foundation, Grant Number 825876, Awardee Thu D. Nguyen.}\\  Rutgers University\\ \texttt{karthik.cs@rutgers.edu}}
\date{}

\usepackage{setspace}
\begin{document}

\maketitle
\begin{abstract}

    The Ulam distance of two permutations on $[n]$ is $n$ minus the length of their longest common subsequence. In this paper, we show that for every $\varepsilon>0$, there exists some $\alpha>0$, and an infinite set $\Gamma\subseteq \mathbb{N}$, such that for all  $n\in\Gamma$, there is an explicit set $C_n$ of $(n!)^{\alpha}$ many permutations on $[n]$,  such that every pair of permutations in $C_n$ has pairwise Ulam distance  at least $(1-\varepsilon)\cdot n$.\vspace{0.15cm}
    
    Moreover,   we can compute the $i^{\text{th}}$  permutation in $C_n$ in $\poly(n)$ time and can also decode in $\poly(n)$ time, a permutation $\pi$ on $[n]$  to its closest permutation $\pi^*$ in $C_n$, if the Ulam distance of $\pi$ and $\pi^*$ is less than $ \frac{(1-\varepsilon)\cdot n}{4} $.\vspace{0.15cm}

    Previously, it was implicitly known by combining works of Goldreich and Wigderson [Israel Journal of Mathematics'23] and Farnoud, Skachek, and   Milenkovic [IEEE Transactions on Information Theory'13] in a black-box manner, that it is possible to explicitly construct $(n!)^{\Omega(1)}$ many permutations on $[n]$,  such that every pair of them have pairwise Ulam distance  at least $\frac{n}{6}\cdot (1-\varepsilon)$, for any $\varepsilon>0$, and the bound on the distance can be improved to $\frac{n}{4}\cdot (1-\varepsilon)$ if the construction of Goldreich and Wigderson is directly analyzed in the Ulam metric. 
\end{abstract}
\clearpage

\section{Introduction}
Permutation codes, pioneered by Slepian in 1965 \cite{slepian1965permutation}, constitute a class of error correction codes extensively explored in both Combinatorics and Information Theory \cite{cameron2010permutation, chee2012efficient,hunt2015decoding}. They have garnered significant attention for their relevance in applications such as Flash memory \cite{jiang2010correcting,jiang2010rank,tamo2010correcting} and Power-line communication \cite{blake1979coding,colbourn2004permutation}. 
Over the past decade, several works have studied Permutation codes in the Ulam metric, and that is the focus of study in this paper.

 The Ulam distance of two permutations on $[n]$ is $n$ minus the length of their longest common subsequence.  We say that a set $S$ of permutations on $[n]$ is a code in the Ulam metric with \emph{distance} $\Delta$ and \emph{rate} $R$ if the minimum Ulam distance between two distinct permutations in \(S\) is $\Delta$  and $R=\frac{\log |S|}{\log (n!)}$. The \textit{relative distance} of   \(S\) is simply $\frac{\Delta}{n}$. 
 
A formal investigation of codes in the Ulam metric was started by Farnoud, Skachek, and
 Milenkovic \cite{farnoud2013error} (although aspects of it were studied by \cite{levenshtein1992perfect,beame2009longest}). Since then, additional  follow-up works have appeared on codes in the Ulam metric \cite{golouglu2015new,hassanzadeh2014multipermutation,kong2016nonexistence}, but the explicit construction of a \emph{good code} in the Ulam metric (i.e., codes with positive constant rate and positive constant relative distance) remained elusive. On that front, it was proved in  \cite[Corollary 26]{farnoud2013error} that a good code over permutations in the Hamming metric can be modified to obtain a good code in the Ulam metric. Goldreich and Wigderson \cite{goldreich2020constructing} recently constructed such an explicit good code in the Hamming metric. 
 
 However, the question of explicitly constructing codes in the Ulam metric with positive constant rate and relative distance approaching 1 remained open. If we uniformly and independently sample $(n!)^{\Omega(1)}$ many permutations on $[n]$, then by noting that the distance of two random permutations is $n-\Theta(\sqrt{n})$ \cite{erdos1935combinatorial,baer1968natural}, and applying the  Chernoff bound, we obtain codes in the Ulam metric with positive constant rate and relative distance arbitrarily close to 1. 
 But can we derandomize the above sampling process by providing an explicit construction? 

 High distance codes (with positive constant rate) have numerous applications in coding theory (for example in code concatenation), and also in theoretical computer science, such as in areas like sketching and hardness of approximation.

\subsection{Our Contribution}

Our main contribution is the construction of infinite families of permutation codes in the Ulam metric with relative distance arbitrarily close to 1 and rate (depending on the relative distance) bounded away from zero.

\begin{theorem}
\label{thm:main}
    Let \(q\) be a sufficiently large positive integer. There exists a constant \(\epsilon > 0\) (depending on $q$) such that for every $\ell\in\mathbb{N}$ and \(n:= q^{\ell}\), there is an integer $M$ and an injective function  \(S\colon [M] \to \mathcal{S}_n\) (which is a code in the Ulam metric) with the following properties.
    \begin{itemize}
        \item \(M \geq \left(n!\right)^{\epsilon}\), i.e., rate of $S$ is at least $\epsilon$.
        \item The relative distance of $S$ is $\delta_S\ge  \left(1-\frac{1}{\mathrm{poly}(q)}\right)$. 
\item    There is an algorithm that, given \(\x\in [M]\), outputs $S(\x)$ in \(\poly(n)\) time.
\item There is a polynomial time algorithm that takes a permutation \(\pi\in \mathcal{S}_n\) as input, and outputs \(\x^*\in [M]\) (if it exists) such that the Ulam distance between $\pi$ and $S(x^*)$ is less than \(\frac{\delta_S}{4} n\).
\end{itemize}
\end{theorem}

In the above theorem, as we increase $q$, we obtain codes with relative distance approaching 1 while the rate is still a positive constant that depends on $q$ (and we treat $q$ as a large constant). We remark here that since the Hamming distance between two permutations is always at least as large as their Ulam distance, the above result also gives the first construction of permutation codes in the Hamming metric to achieve positive constant rate and relative distance arbitrarily close to 1.

To motivate the construction in Theorem~\ref{thm:main}, we describe now a random construction procedure to sample a random permutation. Let $Q_N$ be the boolean hypercube on vertex set $\{0,1\}^N$. Imagine placing $2^N$ balls numbered by the set $\{0,\ldots , 2^N-1\}$ on the vertices of $Q_N$ based on their binary representation (thus there is a ball placed on each vertex). We then begin an $N$-stage shuffling process. At stage $i$, we only consider the edges of $Q_N$ obtained by changing the $i^{\text{th}}$ bit (i.e., the pair of vertices sharing the edge differ on the $i^{\text{th}}$ bit). For each of those edges, we independently and uniformly at random decide whether to swap the balls placed on the two endpoints of the edge. We do the above shuffling process sequentially starting from $i=1$ and ending at $i=N$. The final permutation generated by this process is done by reading the number on the balls located at $00\cdots00, $ followed by $00\cdots01$, then $00\cdots10$, and so on until we reach $11\cdots 11$ (i.e., in increasing order based on binary representation of the vertices). See Figure~\ref{fig1} for an illustration of this process.

\begin{figure}[h]
    \centering
\resizebox{\textwidth}{!}{\includegraphics[scale=0.9]{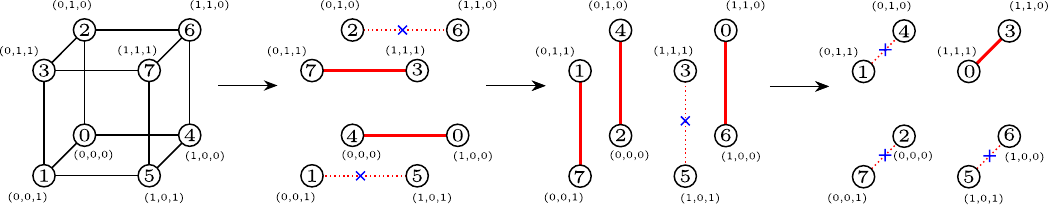}}    \caption{In the figure, we consider setting where $N=3$. At stage 0, we have the identity permutation $(0,1,2,3,4,5,6,7)$. In stage 1, we  decide to swap balls on the edge $000$ and $100$ and on the edge $011$ and $111$ (while not swapping the balls on the edge $010$ and $110$ and the edge $001$ and $101$). Thus, we obtain the permutation $(4,1,2,7,0,5,6,3)$. In stage 2, we decide to swap balls on the edge $000$ and $010$, the edge $001$ and $011$,  and on the edge $100$ and $110$, obtaining the permutation $(2,7,4,1,6,5,0,3)$. Finally, in stage 2, we decide to swap balls only on the edge $110$ and $111$,  obtaining the permutation $(2,7,4,1,6,5,3,0)$ which is our output. Note that all decisions were obtained as outcomes of some coin tosses. }
    \label{fig1}
\end{figure}

One may view the construction of good permutation codes (in the Hamming metric) in \cite{goldreich2020constructing} as a derandomization of the above process, wherein instead of randomly choosing which edges to swap at each stage, the choice is done through a codeword from a good code in the Hamming metric on the Boolean alphabet. 
To the best of our knowledge, either combining \cite[Corollary 26]{farnoud2013error} with the construction in \cite{goldreich2020constructing} or directly analyzing the construction in \cite{goldreich2020constructing} cannot yield positive rate codes in the Ulam metric of relative distance higher than $1/4$. Thus, we need to generalize the above process over larger alphabets while also finding an analogue to swapping that works in the generalized setting.   We discuss more about the comparison of our construction with  \cite{goldreich2020constructing} in Section~\ref{sec:gw23}.

\subsection{Related Works}
The earliest result for codes in the Ulam metric can be found in~\cite{levenshtein1992perfect}, wherein they describe a way of constructing a set of permutations of \([n]\) capable of correcting a single deletion. This, in our language, is equivalent to a code in the Ulam metric with distance 2. Furthermore, this code has size \((n-1)!\), which is optimal for this distance.
In 
\cite{beame2009longest} the authors construct \(O(\log n)\) permutations with pairwise Ulam distance roughly equal to \((n-\sqrt{n})\). Asymptotically speaking, this is a rate 0, distance 1 code.
In \cite{farnoud2013error}, the authors systematically study codes in Ulam metric proving analogues of the Gilbert-Varshamov bound and Singleton bound for the Ulam metric. Moreover, they provide construction of high rate codes albeit with $o(1)$ relative distance. In~\cite{golouglu2015new}, the authors give better upper and lower bounds on the size of codes in the Ulam metric using techniques from  integer programming and also show that the Singleton-type bound derived in~\cite{farnoud2013error} is not tight. In~\cite{hassanzadeh2014multipermutation}, the study of permutation codes in the Ulam metric is extended to multipermutations. In
\cite{kong2016nonexistence} the authors prove the non-existence of perfect codes in the Ulam metric. Finally, in \cite{goldreich2020constructing}, the authors give the first explicit construction of asymptotically good permutation codes in the Hamming metric. 

\section{Notations}
In this section, we detail some notations that will be used in the rest of the paper.

For a positive integer \(n\), we denote  \([n]:=\{0, 1, 2, \ldots, n-1\}\). We will denote by \(\mathcal{S}_n\) the set of all length \(n\) strings over \([n]\) with no repeating letters\footnote{The choice of notations here is intentionally made to coincide with the notation for the symmetric group of order $n$, although we never use the group operator in this paper.}. Such strings can be viewed as permutations of the alphabet \([n]\), and we will often call them as such, although we  treat them merely as strings throughout this paper. We refer to the string \(\Id\in \mathcal{S}_n\), where for all \(i\in [n]\), \(\Id[i] = i\), as the \textit{identity permutation}. 


The Hamming distance between two equal-length strings \(x\) and \(y\), denoted by \(d_H(x, y)\), is the number of locations where \(x\) and \(y\) have different symbols. We will use \(L(x, y)\) to denote the length of a longest common subsequence of \(x\) and \(y\). Central to our discussion is the notion of the \textit{Ulam distance} between two permutations. Given \(\pi, \pi'\in \mathcal{S}_n\), the Ulam distance between \(\pi\) and \(\pi'\), denoted by \(d_U(\pi, \pi')\), is the least number of symbol relocations required to transform \(\pi\) into \(\pi'\). Something that we make use of many times throughout the paper is the fact~\cite{aldous1999longest} that for every \(\pi, \pi'\in \mathcal{S}_n\), we have \(d_U(\pi, \pi') = n - L(\pi, \pi')\).

\section{Construction of the Code}

\label{sec:construction}

In this section, for a fixed $q,\ell,$ and $n=q^\ell$, we construct the function $S$ specified in Theorem~\ref{thm:main}. It is more convenient to present the construction and analysis if we think of $S$ through its image set which is a subset of $\mathcal{S}_n$, and thus we switch to referring to $S$ as a subset of strings in $\mathcal{S}_n$. 
To describe the elements in \(S\), we need two main ingredients. The first is a \textit{ground permutation set} \(D=\{\sigma_0, \sigma_1, \ldots, \sigma_{p-1}\}\) of \(p\) permutations of \([q]\) for some integer \(p\). The second is a set \(C\) of strings over \([p]\) of length \(\frac{n}{q}\), referred to hereafter as the \textit{shuffler set}. Although our construction can be carried out with any choice of the sets \(D\) and \(C\), it is helpful to think of the set \(D\) as a fixed-size high-distance code in the Ulam metric and the set \(C\) as an asymptotically good \(p\)-ary code in the Hamming metric (the exact choice of parameters will be specified in Section~\ref{sec:params}). Our construction serves the intent of using the properties of \(C\) to ``lift'' the distance of the fixed-size code \(D\) to obtain a family of high-distance good codes in the Ulam metric.

Given the sets \(D\) and \(C\), we now construct \(S\) as follows. We view each string in \(S\) as being generated by an \(\ell\)-stage process. The process begins with the identity permutation \(\pi^{(0)}:=\Id\), and at the end of stage \(i\in\{1, 2, \ldots, \ell\}\), we generate the permutation \(\pi^{(i)}\) from the permutation  \(\pi^{(i-1)}\) generated at the end of the previous stage. The mechanism of generating \(\pi^{(i)}\) from \(\pi^{(i-1)}\) is dictated by some string \(w^{(i)} \in C\), which we call the \textit{shuffler} for stage \(i\).
Intuitively,  each symbol of the shuffler specifies how to shuffle the contents of some group of locations in \(\pi^{(i-1)}\), and we introduce below a few sentences of formalism before we specify the precise description. 


Consider associating a string \(\phi_q(m)\in [q]^\ell\) with each \(m\in [n]\), where \(\phi_q(m)\) is the base-\(q\) representation of \(m\). 
In what follows, we will index permutations of length $n$ using the set of strings in $[q]^{\ell}$. In other words, for every \(\pi\in \mathcal{S}_n\) and \(m\in [n]\), we will write \(\pi[\phi_q(m)]\) to denote the symbol at location \(m\) in \(\pi\).

We are now ready to describe the details of our construction. Fix some \(i\in\{1, 2, \ldots, \ell\}\). 
 During the \(i^{\text{th}}\) stage, for every $(\alpha,\beta)\in [q]^{i-1}\times    [q]^{\ell-i}$, we first consider groups of indices/locations  \(I_{\alpha, \beta}^{(i)}:=\{m: \phi_q(m)=\alpha x \beta, \text{ for some }x\in[q]\}\subseteq [n]\). Note that there are precisely \(\frac{n}{q}\) such groups, which is the same as the length of the shuffler \(w^{(i)}\). Next, we associate each group \(I_{\alpha, \beta}^{(i)}\) with a unique location in \(w^{(i)}\) in some canonical way.  We call this location \((\alpha, \beta)\) and use \(w^{(i)}[(\alpha, \beta)]\) to denote the symbol residing there. Finally, for each group \(I_{\alpha, \beta}^{(i)}\), we shuffle its contents according to the ground permutation \(\sigma_{w^{(i)}[(\alpha, \beta)]} \in D\), where \(w^{(i)}[(\alpha, \beta)]\in[p]\) is the symbol at location \((\alpha, \beta)\) in \(w^{(i)}\). 
 More formally, for each \((\alpha, \beta)\in [q]^{i-1}\times[q]^{\ell-i}\) and \(x\in [q]\), we set: \[\pi^{(i)}[\alpha  x  \beta] := \pi^{(i-1)}[\alpha  y\beta],\quad\text{where}\quad y:=\sigma_{w^{(i)}[(\alpha,\beta)]}[x].\] 
 We thus obtain \(\pi^{(i)}\) from \(\pi^{(i-1)}\). This concludes the \(i^{\text{th}}\) stage. The process ends after stage \(\ell\), and we end up with some permutation \(\pi^{(\ell)}\in \mathcal{S}_n\) depending on the sequence of shufflers used.

For every $\ell$-tuple of strings  in $C$, i.e., \(\w=(w^{(1)}, w^{(2)}, w^{(3)}, \ldots, w^{(\ell)}) \in C^{\ell}\), we denote by \(\pi_{\w}\) the permutation obtained after stage \(\ell\) in this process by using the \(w^{(i)}\)'s as shufflers. The set \(S\) we consider is simply the set of all permutations that can be generated by each \(\ell\)-tuple of shufflers. In other words, we consider:
\[S=\{\pi_{\w} : \w\in C^{\ell}\}.\]

\begin{figure}
    \centering
\resizebox{0.7\textwidth}{!}{\includegraphics[scale=0.9]{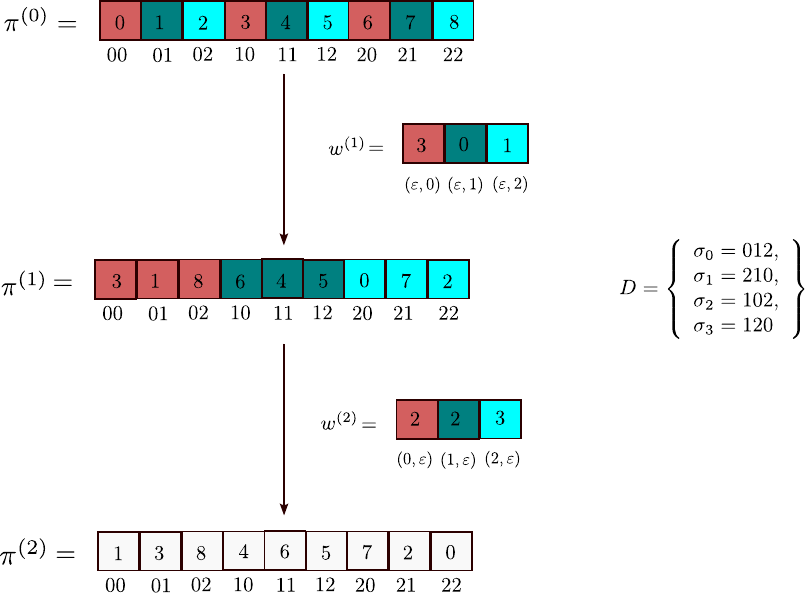}}    \caption{We illustrate in this figure the construction of a codeword described in Section~\ref{sec:construction}, for the setting where $q=3, \ell=2, n=9, p=4, D = \{\sigma_0=012,\sigma_1=210, \sigma_2=102, \sigma_3=120 \},$ and $ \w=(w^{(1)}=301, w^{(2)}=223)$.}
    \label{fig2}
\end{figure}

We have provided in Figure~\ref{fig2}, an illustration of the above construction for $q=3$ and $\ell=2$.

\section{Parameters of the Code}
\label{sec:params}
In this section, we prove the rate, distance, and encoding claims of Theorem~\ref{thm:main}. 
To achieve the guarantees claimed in Theorem~\ref{thm:main}, we choose \(D\) and \(C\) with certain parameters. First, we construct $D$ such  that \(p=\lvert D \rvert \geq q^{\epsilon_D}\) for some constant \(\epsilon_D\in (0, 1)\) and for every distinct \(\sigma_i, \sigma_j \in D\), we have \(L(\sigma_i, \sigma_j) \leq q^{\delta_D}\) for some constant \(\delta_D\in (0, 1)\). The existence of such a set \(D\) with the above properties is always guaranteed due to Gilbert-Varshamov bound type arguments \cite[Proposition 7]{farnoud2013error}, and moreover, since \(q\) is a constant, one can simply brute force all permutations of \([q]\) to find such a \(D\). Later in Section~\ref{sec:ground-set-construction}, we even show how to construct $D$ explicitly without resorting to brute-force search.

Second, we construct $C$ such that \(\lvert C \rvert \geq \left(p^{\frac{n}{q}}\right)^{R_C}\) for some constant \(R_C>0\), and each pair of strings in \(C\) has relative Hamming distance at least \(\delta_C\) for some constant \(\delta_C>0\). This can be done by letting \(C\) be the set of codewords of some asymptotically good code (in the Hamming metric) over alphabet \([p]\) with rate and relative distance \(R_C\) and \(\delta_C\), respectively. With such sets \(D\) and \(C\) in hand, we make the following claim.
\begin{claim}
\label{claim:main}
    Let \(D\) and \(C\) be sets satisfying the properties above and let \(S\) be defined as in Section~\ref{sec:construction}. Then the following holds.
    \begin{itemize} 
        \item \(\lvert S \rvert \geq \lvert \mathcal{S}_n \rvert^{\frac{\epsilon_DR_C}{q}}\),
        \item For distinct \(\pi, \pi'\in S\), \(L(\pi, \pi') \leq \left(\frac{\delta_C}{q^{1-\delta_D}}+(1-\delta_C)\right)n\).
    \end{itemize}
\end{claim}
\begin{proof}
    We first prove the second claim. Let \(\w, \mathbf{w'} \in C^\ell\) such that \(\w\neq \mathbf{w'}\) and consider the pair of strings \(\pi = \pi_{\w}\) and \(\pi' = \pi_{\mathbf{w'}}\). For \(1\leq i \leq \ell\), let \(\pi^{(i)}\) and \(\pi'^{(i)}\) be the strings obtained at the end of stage \(i\) when using \(\w\) and \(\mathbf{w'}\), respectively, as the sequence of shufflers. By this notation, \(\pi^{(\ell)} =\pi\) and \(\pi'^{(\ell)} = \pi'\). The key observation is that if different shufflers are used during some stage, it forces the resulting strings to have small longest common subsequences. Furthermore, this can not be undone during subsequent stages -- the length of a longest common subsequence can only go down with each stage.

    We now make this observation precise. For any fixed subset \(A\subseteq [n]\), denote by \(\pi|_{A}\) the string obtained by deleting from \(\pi\) all the symbols in \([n]\setminus A\).  Let  \(\w=(w^{(1)}, w^{(2)}, \ldots, w^{(\ell)})\),  \(\mathbf{w'}=({w'}^{(1)}, {w'}^{(2)}, \ldots, {w'}^{(\ell)})\), and consider the smallest integer \(j\) such that \(w^{(j)} \neq {w'}^{(j)}\). Our first step is to bound \(L(\pi^{(j)}, \pi'^{(j)})\) using the following simple observation. 
    \begin{obs}
    \label{obs:restrict}
        Let $A_1\dot\cup \cdots \dot\cup A_k$ be a partition of \([n]\) for some $k\in\mathbb{N}$. Then for every pair of strings \(\pi, \pi' \in \mathcal{S}_n\), we have \(L(\pi, \pi')\leq \sum_{i=1}^k L(\pi|_{A_i}, \pi'|_{A_i})\).
    \end{obs}
    \begin{proof}
        Let \(\rho\) be a longest common subsequence of \(\pi\) and \(\pi'\). Then \(L(\pi, \pi') = \lvert \rho \rvert = \sum_{i=1}^k\lvert\rho|_{A_i}\rvert \leq \sum_{i=1}^k L(\pi|_{A_i}, \pi'|_{A_i})\).
    \end{proof}
    Now for each \((\alpha, \beta) \in [q]^{j-1}\times [q]^{\ell-j}\), define the set \(A_{\alpha, \beta}=\{\pi^{(j-1)}[\alpha  x  \beta]=\pi'^{(j-1)}[\alpha  x  \beta] = \phi_q^{-1}(\alpha  x  \beta) : x\in [q]\}\). Clearly, the sets \((A_{\alpha, \beta})_{(\alpha, \beta) \in [q]^{j-1}\times [q]^{\ell-j}}\) form a partition of \([n]\). Therefore, by Observation~\ref{obs:restrict}, we have \[L\left(\pi^{(j)}, \pi'^{(j)}\right) \leq \sum_{(\alpha, \beta) \in [q]^{j-1}\times [q]^{\ell-j}} L\left(\pi^{(j)}|_{A_{\alpha, \beta}}, \pi'^{(j)}|_{A_{\alpha, \beta}}\right).\]
    Furthermore, we have \(L\left(\pi^{(j)}|_{A_{\alpha, \beta}}, \pi'^{(j)}|_{A_{\alpha, \beta}}\right) = L\left(\sigma_{w^{(j)}[(\alpha, \beta)]}, \sigma_{w'^{(j)}[(\alpha, \beta)]}\right)\). So, if \(w^{(j)}[(\alpha, \beta)] \neq w'^{(j)}[(\alpha, \beta)]\), then \(L\left(\pi^{(j)}|_{A_{\alpha, \beta}}, \pi'^{(j)}|_{A_{\alpha, \beta}}\right)\leq q^{\delta_D}\). Otherwise, \(L\left(\pi^{(j)}|_{A_{\alpha, \beta}}, \pi'^{(j)}|_{A_{\alpha, \beta}}\right) = q\). Since \(w^{(j)}\) and \(w'^{(j)}\) are distinct codewords of a code with relative distance \(\delta_C\), we have:
    \begin{align*}
      L\left(\pi^{(j)}, \pi'^{(j)}\right) &\leq \sum_{(\alpha, \beta) \in [q]^{j-1}\times [q]^{\ell-j}} L\left(\pi^{(j)}|_{A_{\alpha, \beta}}, \pi'^{(j)}|_{A_{\alpha, \beta}}\right)\\
      &=\sum_{\substack{(\alpha, \beta) \in [q]^{j-1}\times [q]^{\ell-j}\\w^{(j)}[(\alpha, \beta)] \neq w'^{(j)}[(\alpha, \beta)]}} L\left(\pi^{(j)}|_{A_{\alpha, \beta}}, \pi'^{(j)}|_{A_{\alpha, \beta}}\right)+\sum_{\substack{(\alpha, \beta) \in [q]^{j-1}\times [q]^{\ell-j}\\w^{(j)}[(\alpha, \beta)] = w'^{(j)}[(\alpha, \beta)]}} L\left(\pi^{(j)}|_{A_{\alpha, \beta}}, \pi'^{(j)}|_{A_{\alpha, \beta}}\right)\\
      &\leq \delta_C \cdot \frac{n}{q}\cdot q^{\delta_D} + (1-\delta_C)\cdot\frac{n}{q}\cdot q\\
      &=\left(\frac{\delta_C}{q^{1-\delta_D}}+(1-\delta_C)\right)n
    \end{align*}
    Furthermore, since for any \(m\in [n]\), the \(j^{\text{th}}\) symbol of the base-\(q\) representation of the location of \(m\) is forever fixed after stage \(j\), during any stage \(k\) with \(k>j\), the relative order of the symbols of \(A_{\alpha, \beta}\) in \(\pi^{(k)}\) and \(\pi'^{(k)}\) does not change with respect to \(\pi^{(j)}\) and \(\pi'^{(j)}\), respectively, for any \((\alpha, \beta) \in [q]^{j-1}\times [q]^{\ell-j}\). In other words, \(L\left(\pi^{(k)}|_{A_{\alpha, \beta}}, \pi'^{(k)}|_{A_{\alpha, \beta}}\right)=L\left(\pi^{(j)}|_{A_{\alpha, \beta}}, \pi'^{(j)}|_{A_{\alpha, \beta}}\right)\) for every \(k>j\) and \((\alpha, \beta) \in [q]^{j-1}\times [q]^{\ell-j}\). The claim then follows, and we have \(L(\pi, \pi') \leq \left(\frac{\delta_C}{q^{1-\delta_D}}+(1-\delta_C)\right)n\).

    It only remains to bound the size of \(S\). Note that: \[\lvert S \rvert = \lvert C^{\ell}\rvert = \left(p^{\frac{nR_C}{q}}\right)^{\log_q n}= \left((q^{\epsilon_D})^{\frac{nR_C}{q}}\right)^{\log_q n} = (n^n)^{\frac{\epsilon_D R_C}{q}} \geq (n!)^{\frac{\epsilon_D R_C}{q}} = \lvert \mathcal{S}_n \rvert^{\frac{\epsilon_D R_C}{q}}.\qedhere\]\end{proof}

Equipped with Claim~\ref{claim:main}, we can now prove the rate, distance, and encoding claims of Theorem~\ref{thm:main}. To prove the rate claim, we just set \(\epsilon =  \frac{\epsilon_D R_C}{q}\) and we have $M=|S|\ge  (n!)^{\epsilon}$. 

Since Ulam distance is \(n\) minus the longest common subsequence length, for distinct \(\pi, \pi'\in S\), we have \(d_U(\pi, \pi') \geq \delta_C\left(1-\frac{1}{q^{1-\delta_D}}\right)n\). To match the \(\left(1-\frac{1}{\mathrm{poly}(q)}\right)n\) distance lower bound, we can make use of AG codes which let us achieve \(\delta_C \geq 1-\frac{3}{\sqrt{q^{\epsilon_D}}}\) while still keeping \(R_C\) at least a constant~\cite{garcia1996asymptotic} (in particular, we can have \(R_C\ge \frac{1}{\sqrt{q^{\epsilon_D}}}\)). The use of AG code does require \(p\) to be a prime square greater than or equal to 49. However, this issue can be handled without affecting the rate and distance guarantees in the following way. 
First we pick \(q\) large enough so that \(q^{\frac{\epsilon_D}{2}} \geq 49\). Now if \(p \geq q^{\epsilon_D}\) is already a prime square, we are done. Otherwise, we find a prime square \(p' \in [p/4, p]\), which is guaranteed to exist as a consequence of Bertrand's postulate. We then use AG codes over \([p']\) for our shuffler set \(C\). Note that \(p' \geq \frac{p}{4} \geq \frac{q^{\epsilon_D}}{4} \geq q^{\frac{\epsilon_D}{2}}\). So, using \(p'\) instead of \(p\) can only make the rate of our code go down by at most a factor of 2 while still keeping the relative distance at least \(\left(1-\frac{1}{\mathrm{poly}(q)}\right)\).

Finally, to see that for every \(\x\in [M]\) we can  output $S(\x)$ in \(\poly(n)\) time, note from the construction in Section~\ref{sec:construction} that the run time is primarily determined by the time needed to encode $\x$ in order to obtain $\w\in C^{\ell}$. This encoding is possible because we can first assume that $x$ is a vector in $\left[M^{1/\ell}\right]^{\ell}$, denoted by $(x_1,\ldots ,x_{\ell})$, and consider some canonical ordering of   $\left[p\right]^{\frac{nR_C}{q}}$, denoted by $\theta\colon \left[p^{\frac{nR_C}{q}}\right]\to\left[p\right]^{\frac{nR_C}{q}} $. Then for each $x_i\in \left[p^{\frac{nR_C}{q}}\right]$ (where $i\in \{1, 2, \ldots , \ell\}$),  we construct the AG code encodings of $\theta(x_1)$, \ldots, $\theta(x_\ell)$ to obtain an element in $C^{\ell}$. 
Since AG codes can be encoded in time which is polynomial in the block length \cite{SAKSD01}, we can thus output $S(\x)$ in \(\poly(n)\) time.

\subsection{Explicit Construction of the Ground Permutation Set}
\label{sec:ground-set-construction}
Here we show how to explicitly construct the ground permutation set \(D\) without resorting to brute-force search. Recall that the ground permutation set \(D\subseteq \mathcal{S}_q\) is a set of permutations such that the following two conditions hold:
\begin{itemize}
    \item \(\lvert D \rvert \geq q^{\epsilon_D}\) for some constant \(\epsilon_D \in (0, 1)\) independent of \(q\).
    \item For distinct \(\sigma, \sigma'\in D\), \(L(\sigma, \sigma')\leq q^{\delta_D}\) for constant \(\delta_D\in (0, 1)\) independent of \(q\).
\end{itemize}

The construction we give is largely inspired by~\cite{beame2009longest}. Assume \(q\) to be a power of two and let \(G\subseteq \{0, 1\}^{\log_2 q}\) be a code with rate \(\epsilon_D\) and relative Hamming distance \((1-\delta_D)\). For each \(g\in G\), we define the permutation \(\sigma_g\in \mathcal{S}_q\) as follows. For each \(x\in [q]\), we set \(\sigma_g[i]=x\) if and only if \(i\) equals the bitwise XOR of \(x\) and \(g\). Finally, we set \(D=\{\sigma_g\in \mathcal{S}_q : g\in G\}\).

It is immediately clear that \(\lvert D \rvert = \lvert G \rvert \geq \left(2^{\log_2 q}\right)^{\epsilon_D}=q^{\epsilon_D}\). Assume for the sake of contradiction that there exist distinct \(\sigma_g, \sigma_{g'}\in D\) such that \(L(\sigma_g, \sigma_{g'}) > q^{\delta_D}\). First we make the following observation.
\begin{obs}
\label{obs:ground}
    Let \(g\in G\) and \(x, y\in [q]\). If \(m\in [\log_2 q ]\) is the smallest integer for which \(\phi_2(x)[m] \neq \phi_2(y)[m]\) (where $\phi_2$ maps to the binary representation), then the relative order of \(x\) and \(y\) in \(\sigma_g\) is completely determined by the bit \(g[m]\).
\end{obs}
Now let \(I\subseteq [\log_2 q]\) be the set of locations \(i\) such that \(g[i]=g'[i]\). Since \(G\) has relative distance \((1-\delta_D)\), we have \(\lvert I \rvert \leq \delta_D\log_2 q\). Since \(L(\sigma_g, \sigma_{g'}) > q^{\delta_D}\), by the pigeonhole principle, there exist distinct \(x, y\) in the longest common subsequence of \(\sigma_g\) and \(\sigma_{g'}\) such that \(\phi_2(x)[i]=\phi_2(y)[i]\) for each \(i\in I\). However, note that if \(m\in [\log_2 q]\) is the smallest integer for which \(\phi_2(x)[m] \neq \phi_2(y)[m]\), then \(m\notin I\). So, by Observation~\ref{obs:ground}, the relative order of \(x\) and \(y\) is different in \(\sigma_g\) and \(\sigma_{g'}\), and as such, they can not be in a common subsequence of \(\sigma_g\) and \(\sigma_{g'}\), which is a contradiction. Therefore, \(L(\sigma_g, \sigma_{g'}) \leq q^{\delta_D}\) for all distinct \(\sigma_g, \sigma_{g'}\in D\).

\subsection{Comparison with the Construction of Goldreich and Wigderson}\label{sec:gw23}

Our construction can be viewed as a generalization of the construction of good permutation codes in the Hamming metric given by Goldreich and Wigderson~\cite{goldreich2020constructing}. In their construction, certain locations in the permutation are paired and then either swapped or not swapped. This can be viewed as setting both \(p=q=2\) and letting the ground permutation set \(D\) consist of the two possible permutations of \(\{0, 1\}\) in our construction. We note that by using the same arguments as in the proof of Claim~\ref{claim:main}, one can show that the Goldreich-Wigderson construction without any modification already gives asymptotically good permutation codes not only in the Hamming metric but also in the Ulam metric. However, setting \(p=q=2\) restricts the code from having relative Ulam distance higher than \(\frac{1}{4}\). In our construction, \(p\neq q\) and neither of them is necessarily equal to 2. Furthermore, our ground permutation set consists of permutations with relative Ulam distance \((1-o(1))\). Both of these contribute to achieving codes with constant rate and relative Ulam distance approaching 1.

\section{Efficient Decoding of the Code}

In this section, we prove the decoding claim of Theorem~\ref{thm:main}.
    Let \(\pi\in \mathcal{S}_n\) such that there exists a unique \(\pi^*\in S\) with \(d_U(\pi, \pi^*) < \frac{\delta_S}{4}n\). We show how to recover \(\pi^*\) from \(\pi\). Let \(\pi^* = \pi_{\mathbf{w^*}}\), where \(\mathbf{w^*}=(w^{*(1)}, w^{*(2)}, w^{*(3)}, \ldots, w^{*(\ell)})\). Just as before, for \(1\leq i \leq \ell\), we denote by \(\pi^{*(i)}\) the string obtained at the end of stage \(i\) when \(\mathbf{w^*}\) is used as the sequence of shufflers. Our decoding algorithm is also composed of \(\ell\) stages. During the \(i^{\text{th}}\) stage, our goal is to recover the shuffler \(w^{*(i)}\), and, consequently, \(\pi^{*(i)}\), in addition to correcting the relative order of some set of symbols in \(\pi\). We do this by first looking at \(\pi^{*(i-1)}\) (which we already have at the end of stage \((i-1)\)), and forming a partition of \([n]\). For each set in this partition, we then try to guess the relative order of all the symbols in that set in \(\pi^*\). For every set for which we have a correct guess, we are successful in correctly determining one symbol of \(w^{*(i)}\) at some specified location. The key insight is that since \(\pi\) has small Ulam distance to \(\pi^*\), most of our guesses will be correct, and as a result, the shuffler we end up guessing will not be too far from the true shuffler \(w^{*(i)}\). We can then recover the true shuffler \(w^{*(i)}\) by using the decoding algorithm for \(C\). Finally, by using \(w^{*(i)}\), we obtain \(\pi^{*(i)}\), and for each set of the partition previously formed, we correct in \(\pi\) the relative order of all the symbols in that set. Details follow.

    Assume, we are in the \(i^{\text{th}}\) stage and we already have \(\pi^{*(i-1)}\). For each \((\alpha, \beta) \in [q]^{i-1}\times [q]^{\ell-i}\), we first define the set \(A_{\alpha, \beta}=\{\pi^{*(i-1)}[\alpha  x  \beta]: x\in [q]\}\). We then find the permutation \(\sigma_c\in D\) that minimizes the Ulam distance between \(\pi|_{A_{\alpha, \beta}}\) and \(\sigma_c\) applied to \(\pi^{*(i-1)}|_{A_{\alpha, \beta}}\). Since \(\lvert D \rvert\) is constant, this can be done in constant time. We then set \(w^{(i)}[(\alpha, \beta)] =c\). After all the \(w^{(i)}[(\alpha, \beta)]\)s have been set in this way, we run the decoding algorithm for \(C\) on \(w^{(i)}\) to obtain \(w^{*(i)}\). Finally, by using \(w^{*(i)}\), we generate \(\pi^{*(i)}\) from \(\pi^{*(i-1)}\), and for each \((\alpha, \beta) \in [q]^{i-1}\times [q]^{\ell-i}\), shuffle the symbols in \(\pi\) so that \(\pi|_{A_{\alpha, \beta}} = \pi^{*(i)}|_{A_{\alpha, \beta}}\). This concludes the \(i^{\text{th}}\) stage.

  Note that $C$ is simply an AG code (as described in the previous section) and thus, the decoding algorithm is correct if \(d_H(w^{(i)}, w^{*(i)}) < \frac{\delta_C}{2}\cdot \frac{n}{q}\) (this is known from \cite{SAKSD01}; also see Guruswami's appendix in \cite{shpilka2009constructions}).  Recall that \(d_U(\pi, \pi^*) < \frac{\delta_S}{4}n =\frac{\delta_C}{4}\left(1-\frac{1}{q^{1-\delta_D}}\right)n\). By Observation~\ref{obs:restrict}, \(d_U(\pi, \pi^*)\geq \sum_{(\alpha, \beta) \in [q]^{i-1}\times [q]^{\ell-i}} d_U\left(\pi|_{A_{\alpha, \beta}}, \pi^*|_{A_{\alpha, \beta}}\right)\).  Call a pair \((\alpha, \beta)\in [q]^{i-1}\times [q]^{\ell-i}\) \textit{good} if \(d_U\left(\pi|_{A_{\alpha, \beta}}, \pi^*|_{A_{\alpha, \beta}}\right) < \frac{q}{2}\left(1-\frac{1}{q^{1-\delta_D}}\right)\). We claim that at least \((1-\frac{\delta_C}{2})\)-fraction of the pairs \((\alpha, \beta)\) are good. Assume otherwise, then for more than \(\frac{\delta_C}{2}\)-fraction of the pairs \((\alpha, \beta)\), \(d_U\left(\pi|_{A_{\alpha, \beta}}, \pi^*|_{A_{\alpha, \beta}}\right) \geq \frac{q}{2}\left(1-\frac{1}{q^{1-\delta_D}}\right)\). In total, these pairs will contribute at least \(\frac{\delta_C }{2}\cdot \frac{n}{q}\cdot \frac{q}{2}\left(1-\frac{1}{q^{1-\delta_D}}\right) = \frac{\delta_C}{4}\left(1-\frac{1}{q^{1-\delta_D}}\right)n\) to the sum, which leads to a contradiction.

    The final observation is that the strings in \(D\) have pairwise Ulam distance at least \(q\left(1-\frac{1}{q^{1-\delta_D}}\right)\). So, if a pair \((\alpha, \beta)\) is good, i.e., \(d_U\left(\pi|_{A_{\alpha, \beta}}, \pi^*|_{A_{\alpha, \beta}}\right) < \frac{q}{2}\left(1-\frac{1}{q^{1-\delta_D}}\right)\), then \(w^{(i)}[(\alpha, \beta)]\) is set correctly, i.e., \(w^{(i)}[(\alpha, \beta)] = w^{*(i)}[(\alpha, \beta)]\). Since we have just shown that there are at least \((1-\frac{\delta_C}{2})\frac{n}{q}\) good pairs, it follows that \(d_H(w^{(i)}, w^{*(i)}) < \frac{\delta_C}{2}\cdot \frac{n}{q}\).

\section{Conclusion}
    In this paper, we constructed codes with positive constant rate and relative distance arbitrarily approaching 1. We also provided efficient encoding and decoding algorithms for our code.  

    A natural open question is to improve the parameters of our construction to meet the parameters obtained by random codes in the Ulam metric (via sampling). Another question is to improve the decoding radius of our algorithm.
    It would also be interesting to explore concrete applications of high-distance codes in the Ulam metric to concrete problems in Theoretical Computer Science (for example, in rank aggregation).

\subsection*{Acknowledgements}
We thank Venkatesan Guruswami for helping us with some references. 
\bibliographystyle{alpha}
\bibliography{bibliography.bib}

    

\end{document}